\documentclass{amsart}
\usepackage{amsmath}%
\usepackage{amsfonts}%
\usepackage{amssymb}%
\usepackage{graphicx}
\usepackage{framed}

\usepackage[utf8]{inputenc}
\theoremstyle{plain}

\newtheorem{corollary}{Corollary}

\newtheorem{definition}{Definition}

\newtheorem{lemma}{Lemma}

\numberwithin{equation}{section}
\begin{document}
\title[Dynamic Longest Subsequence]{The Dynamic Longest Increasing Subsequenc Problem}
\author{Alex Chen, Timothy Chu, Nathan Pinsker}%

\maketitle

\begin{abstract}
In this paper, we construct a data structure using a forest of red-black trees to efficiently compute the longest increasing subsequence of a dynamically updated sequence. Our data structure supports a query for the longest increasing subsequence in $O(r+\log n)$ worst case time and supports inserts anywhere in the sequence in $O \left(r\log{n/r}\right)$ worst case time, where $r$ is the length of the longest increasing subsequence. The data structure can be augmented to support $O(\log n)$ worst case time insertions if the insertions are performed at the end of the sequence. The data structure presented can be augmented to support delete operations in the same time as insertions.
\end{abstract}

\section{Introduction}

The following dynamic longest increasing subsequence problem will be the primary focus of our paper:  given an array $A$ of $n$ elements $e_1, ..., e_n$, find an algorithm to support the following two operations:

\begin{itemize}
\item
\textbf{Insert}: Insert an item after an element $e_i$ or before element $e_1$.

\item
\textbf{Delete}: Delete an element $e_i$.

\item
\textbf{Query}: Calculate the length of the longest increasing subsequence of the array A.
\end{itemize}

A \emph{subsequence} of $A$ is the ordered subset of $A$ corresponding to some list of $m$ indices ($m \le n$) $i_1, i_2, ..., i_m$ with $i_k < i_{k+1}$ for all $k$: $e_{i_1}, e_{i_2}, ..., e_{i_m}$. An increasing subsequence satisfies the requirement that $e_{i_k} < e_{i_{k+1}}$ for all $k$.
\\

The longest increasing subsequence is often used as a measure of how close an input sequence is to being fully sorted (``sortedness''). A solution to the dynamic longest increasing subsequence problem would allow us to dynamically maintain the sortedness of a list as it is modified. Sortedness is useful for applications that require an approximately sorted list and allows us to quickly check whether we need to sort the list or whether it is "close enough'' to be used already. For example, in search engines, we often rank web pages. Applications keep a list of billions of changing web pages that are stored in order by some metric. When a page changes, its score may change, breaking the sorted order. However, we do not want to re-sort the entire list after each page change. Instead, we can use the list's sortedness to determine when to re-sort the list \cite{gopalan07}. A solution to the dynamic longest increasing subsequence problem brings the possibility of inserting new web pages and removing some from the list while still measuring sortedness.

\section{Prior Work}

Significant research has been done in more specific versions of the dynamic longest increasing subsequence problem but not in this generalized version that includes both insertions and deletions anywhere in the list.
\\

Much work has been put into the variation where insertions are allowed only at the end of the list and no deletions are allowed. This is effectively the online version of the longest increasing subsequence problem. There exists a well-known solution that uses $O(n)$ space and $O(n \log n)$ time and finds the length of the longest increasing subsequence exactly. Using dynamic programming, we can keep track of the smallest possible last number of a sequence for each possible sequence length. We compute the table $f[l]$, where $f[l]$ is equal to the smallest possible last element of an increasing subsequence of length $l$ and $\infty$ if there are no increasing subsequences of length $l$. We process the elements of $A$ in order and update $f[l]$ accordingly. The elements of $f[l]$, not including any $\infty$ values, are always in strictly increasing order. When a new element $e_i$ is considered, at most one value of $f[l]$ changes, and this changing value can be found by binary search. Each binary search takes $O(\log n)$ time, and we perform one binary search for each element $e_i$ that we process, resulting in an $O(n \log n)$ algorithm. Because this algorithm processes the elements of $A$ in order and uses only one pass through the data, this algorithm works as an online algorithm to compute the answer exactly in $O(n \log n)$ time.
\\

Ergun and Jowhari (2008) worked with online algorithms for approximating the length of the longest increasing subsequence. This is equivalent to only performing inserts at one end of the list and outputting the length of the longest increasing subsequence after each insert. No deletions are considered here. They proved that in order to approximate the length of the longest increasing subsequence to within $(1 + \epsilon)$, $\Omega(\sqrt{n})$ space is required \cite{ergun08}.
\\

Gopalan et al. (2007) also worked with solving the longest increasing subsequence problem when the list given as a data stream. They proved that a lower bound of $\Omega(n)$ space was required to exactly calculate length of the longest increasing subsequence. Thus, approximation is necessary to use sublinear space. This group also presented a $O(\sqrt{n})$-space deterministic algorithm to approximate the length of the longest increasing subsequence to within a factor of $(1 + \epsilon)$ \cite{gopalan07}.

\section{Main result}
We present the following statement, which is equivalent to the dynamic longest increasing subsequence problem defined earlier. The remainder of our paper is devoted to solving this problem efficiently:

\begin{framed}
\textit{Input}: A sequence of pairs of numbers $(i, v_i)$ (fed online) (all $i$'s are distinct).

\textit{Output}: A sequence $i_1 < i_2 < \ldots  < i_k$ with $v_{i_1} < v_{i_2} < \ldots < v_{i_k}$ of maximum length. 

\end{framed}
To avoid the double-subscript notation, we define $f(i):= v_i$.
\begin{definition} 
An increasing sequence of indices $i_1 < i_2 < \ldots  < i_k$ satisfies $f(i_1) < f(i_2) < \ldots f(i_k)$
\end{definition}

\begin{definition} 
$i << j \Longleftrightarrow i<j, f(i) < f(j)$. 
\end{definition}

Note $<<$ defines a partially ordered set on the elements $(i, f(i))$. Basic dynamic topological sorting algorithms provide an $O(n)$ cost per update guarantee. Bounds better than $O(n)$ for the more general dynamic topological sort problem have been found in \cite{PK}, \cite{MNR}, and \cite{AHRSZ}, although the performance guarantee in our paper on the dynamic increasing sequence problem is faster than those known for the general dynamic topological sort algorithm.

For each index $i$, define $l(i)$ to be the length of the longest increasing subsequence ending at $i$. Index $i$ is said to be in \textit{level} $k$ if $l(i) = k$.

The proof for the next three lemmas are straightforward and will be provided in the appendix.
\begin{lemma} 
\label{1}
Let $j$ be the second-last index of a length $l(i)$ sequence ending at $i$. Then $l(j) = l(i)-1$. 
\end{lemma}

\begin{lemma} 
\label{2}
$i << j \Rightarrow l(i) < l(j)$. 
\end{lemma}

\begin{corollary}
\label{c2}
If $i < j$ and $l(i) = l(j)$, then $f(i) \geq f(j)$.
\end{corollary}


\begin{lemma}
\label{3} 
Inserting a pair $(x, f(x))$ increases $l(i)$ by at most $1$. 
\end{lemma}

Define $L_k$ to be the set of indices in level $k$ after some number of inserts. Define $L'_k$ to be the set of indices in level $k$ after an additional insertion of $(x,f(x))$. Let $l(i)$ denote the level of $i$ before insertion, and define $l'(i)$ to be the level of $i$ after insertion. Our algorithm provides an efficient way of updating $L_k$ into $L'_k$ during an insertion for all $k$ in which $L'_k$ is not empty. The sets $L_k$ will be referred to as 'level sets'. 

Note that corollary \ref{c2} tells us a sequence of elements in $L_k$ that are increasing on $i$ are non-increasing on $f(i)$

Let $T_k$ be the set of indices $i$ with $l(i) = k, l'(i)=k+1$. Here, the level of the inserted element $x$ is defined as $l(x):= l'(x) -1$.

\begin{lemma}
\label{keyLemma}
$T_{k+1}$ consists of all indices $i \in L_{k+1}$ such that $\exists j \in T_k$ with $j<< i$. 
\end{lemma}

\begin{proof} 

If such a $j$ exists in $T_k$, appending $i$ to a (post-insertion) sequence of length $k+1$ ending at $j$ gives a sequence of length $k+2$.

Likewise, if $T_{k+1} \subset L_{k+1}$ by definition of $T$. Lemma $1$ tells us that if $l'(i)=k+2$, then the second-last element $j$ of any sequence ending at $i$ must satisfy $l'(j) = k+1$. Likewise, if $l(i) = k+1$, then $l(j) = k$. Therefore if $i \in T_{k+1}$, then there exists $j \in T_k$ with $j << i$. Lemma\ref{3} implies $i \in T_{k+1}$.
\end{proof}

\begin{corollary} 
$L'_{k+1} = \left(L_{k+1} - T_{k+1}\right) \cup T_k$. 
\end{corollary}

\subsection{Finding the Longest Subsequence}

\begin{lemma} 
\label{pred}
Given $i \in L_k$ for $k > 1$, let $j$ be the predecessor of $i$ in $L_{k-1}$. Then $f(j) < f(i)$.
\end{lemma}

\begin{proof}
By lemma\ref{1}, there exists an element $j'$ in $L_{k-1}$ with $j' << i$. Since $j$ is the predecessor of $i$, then $j \geq j'$. Corollary \ref{c2} implies $f(j) \leq f(j')$, so $j << i$ and $f(j) < f(i)$ as desired.
\end{proof}
 
Suppose we have inserted some elements of the form $(i, f(i))$, and have corresponding level sets $L_k$. To insert a new element $(x, f(x))$ and maintain our structures $L_k$, we first find $l(x)$ and define $T_{l(x)}:=x$. Then, using lemma \ref{keyLemma}, we can find $L'_{l(x)+1}$ and $T_{l(x)+1}$ from $L_{l(x)+1}$ and $T_{l(x)}$. We then proceed by induction to find $T_{k+1}$ and $L'_k$ for all $k$.

Define $m$ to be the length of the longest subsequence among our inserted pairs $(i, f(i))$. To extract $m$, first determine the largest value of $k$ for which $L_k$ is non-empty. To extract a maximum increasing subsequence, start at any index $i_m$ in $L_m$ and find its predecessor $i_{m-1}$ in $L_{m-1}$. In general, define $i_{k+1}$ to be the predecessor of $i_k$ (using the regular ordering on indices) in $L_{m}$. By construction, $i_k$ is increasing in $k$. Lemma \ref{pred} guarantees that the sequence $f(i_k)$ is increasing in $k$, and that all $i_r$ for $1 \leq r \leq k$ are well defined.  

Any increasing subsequence of any length can be found in a similar manner (but instead of running a predecessor query on $i$, take any element $j$ in the level one less than the level of $i$ where $j << i$). 

It remains to show that the sets $L_k$ can be maintained efficiently. 

\subsection{Maintaining $L_k$}

Suppose an array of elements of the form $(i, f(i))$ with corresponding level sets $L_k$ are given, and then the element $(x, f(x))$ is inserted online into this list. The next lemma we prove will tell us that $T_k$ is a contiguous subsection of $L_k$. 

\begin{lemma}
\label{T}
$T_{k+1}$ consists of the indices $i$ in $L_{k+1}$ satisfying both $\min(T_k) < i$ and $f(i) > f(\max(T_k))$.
\end{lemma}
\begin{proof}
Corollary\ref{c2} tells us that there exists a number $U_{k+1}$ such that for all $i\in L_{k+1}$,

$$min (T_k) < i, f(i) > f(\max(T_k)) $$
$$\Longleftrightarrow min(T_k) < i < U_{k+1}$$

Now we proceed by induction. Let $i\in L_{k+1}$. Assume $T_k$ consists of all elements in $L_k$ with $\min(T_{k-1}) < i < U_k$. (Note that this is true for the base case when $k = l(x)$).  Let $j$ be the predecessor of $i$ in $T_k$. It will turn out that the set of indices satisfying $\min (T_k) < i, f(i) > f(\max(T_k))$ are the indices $i$ where $f(i) > f(j)$.

The induction hypothesis implies that for any $i \in L_{k+1}$, its predecessor in $T_k$ is the same as its predecessor in $L_k$ for all $\min(T_k) < i < \max(T_k)$. Lemma \ref{keyLemma} in conjunction with lemma \ref{pred} imply that all such indices are in $T_{k+1}$. Lemma \ref{keyLemma} also implies that all indices $i \in L_{k+1}$ with $i >> \max(T_k)$ are in $T_{k+1}$. Corollary \ref{c2} tells us that $f(i)$ is non-increasing on $i$ (within $L_{k+1}$), so

$$\{i \in L_{k+1}: \min(T_k) < i < \max(T_k)\} \cup \{i \in L_{k+1}: i >> \max(T_k)\}$$
$$=\{i \in L_{k+1}: i > \min (T_k), f(i) > f(\max(T_k))\}$$

\end{proof}

\subsection{Constructing a Data Structure} We now construct a data structure that can efficiently maintain our sets $L_k$. We would like to be able to find the lowest index $i \in L_{k+1}$ with $\min(T_k) > i$, the largest index $i \in L_{k+1}$ with $f(i) > f(\max(T_k))$, and $i$'s predecessor in $L_{k}$. 

To do this, we store the indices of $L_k$ in a red-black tree. Note that a red-black tree on $i$ is automatically a red-black tree for $f(i)$ sorted in reverse order. by corollary \ref{c2}. Red-black trees with $t$ elements can handle predecessor queries in worst case $\mathcal{O}(\log t)$ time, and can also handle splits and concatenations in $\mathcal{O}(\log t)$ time \cite{Booth92}. It follows that it takes worst-case $\mathcal{O}(\log t)$ time to split a red-black tree into two red-black trees, one of which contains all the indices between $m$ and $M$, and the other of which contains all remaining elements

If item $(x, f(x))$ is inserted, $l(x)$ can be found by running a successor query on $f(x)$ for each of $L_k$ starting from $k=1$ until a $k'$ is determined where $f(x)$ is not the largest element in $L_k$. This can be sped up with binary search on $k$, but doing so does not necessarily change the worst case time of the search.

\vspace{2 mm}

Observe that $T_k$ and $T_{k+1}$ are always disjoint. Extracting a red-black tree on $T_k$ can be done via lemma \ref{T} by finding the value $U_k$ (via predecessor query on the reverse-sorted red-black tree on $f(i)$ for $ i \in L_{k}$ with key $f(\max(T_{k-1}))$, and extracting the appropriate index). Then extract the red-black tree for $T_k$ by finding the red-black tree on $L_k$ with indices between $\min(T_k)$ and $U_k$. $L'_k$ can be obtained via merging the remaining red-black tree with the red-black tree on $T_{k-1}$. 

\vspace{2 mm}

The total worst case run time for a single insert is $\sum_{k=1}^r \mathcal{O}\left(\log \left|L_k|\right|\right)$ where $r$ is the length of the longest subsequence at insertion; we need to perform multiple red-black tree operations on each level set $L_k$, each of which takes time $\mathcal{O}\left(\log |L_k|\right)$. Since $\sum_{k=1}^r |L_k|=n$ and $\log$ is a convex function, the upper bound on the run time for any given insertion of $(x, f(x))$ is $\mathcal{O}\left(r \log \left(\frac{n}{r}\right)\right)$. Note that this data structure can be augmented to provide an $O(\log n)$  run time guarantee if the online insertion occurs at the end of the sequence, by maintaining a balanced binary search tree over all $k$, with its $k$ key elements equal to the maximum index $i$ for $i \in L_k$; the insertion of any value at the end of the sequence will only affect one level set $L_k$, and it takes worst case $O(\log n)$ time to find that level set and worst-case $O(\log n)$ time to recompute $L_k$. 

Future avenues of exploration may include having an algorithm that runs in better time for large $r$; our algorithm matches the relatively poor topological sort bound when $r = n$, and it may be of interest to find an improvement to the algorithm when $r$ is large. Additionally, it could be of potential interest to further explore the fingering properties of our algorithm, and whether it takes less worst-case time to insert elements close to the end.

\section{Acknowledgements} The authors wish to give credit to Professor D.  Karger and Professor E. Demaine for helpful conversations. Additionally, the authors would like to thank Joshua Alman for his proofreading assistance.

\end{document}